\documentclass[journal]{IEEEtran}
\usepackage{epsfig, graphicx}
\usepackage{tikz}
\usepackage[colorlinks,
linkcolor=black,
anchorcolor=black,
citecolor=black
]{hyperref}
\usepackage{amsmath, amsthm, amssymb}
\usepackage{cite}
\usepackage{url}
\usepackage{diagbox}

\newtheorem{theorem}{Theorem}
\newtheorem{lemma}{Lemma}
\newtheorem{definition}{Definition}

\usepackage{algorithm}
\usepackage{algorithmic}

\begin{document}

\title{\LARGE On Mean Absolute Error for Deep Neural Network Based Vector-to-Vector Regression}

\author{Jun Qi,~\IEEEmembership{Student Member,~IEEE},
 	 Jun Du,~\IEEEmembership{Member,~IEEE},
	 Sabato Marco Siniscalchi,~\IEEEmembership{Senior Member,~IEEE},\\
	 Xiaoli Ma,~\IEEEmembership{Fellow,~IEEE},
         and~Chin-Hui Lee,~\IEEEmembership{Fellow,~IEEE}
\thanks{J. Qi, X. Ma and C.-H. Lee are with the School of Electrical and Computer Engineering, Georgia Institute of Technology, Atlanta,
GA, 30332 USA e-mail: (qij41@gatech.edu, xiaoli@gatech.edu, chl@ece.gatech.edu).}
\thanks{J. Du is with the National Engineering Laboratory for Speech and Language Information Processing, University of Science and Technology of China, Hefei 230027, China (e-mail: jundu@ustc.edu.cn).}
\thanks{S. M. Siniscalchi is with the Faculty of Architecture and Engineering, University of Enna ``Kore'', Enna 94100, Italy, and also with the Georgia Institute of Technology, Atlanta, GA 30332 USA (e-mail: marco.siniscalchi@unikore.it).}
}

\maketitle

\begin{abstract}
In this paper, we exploit the properties of mean absolute error (MAE) as a loss function for the deep neural network (DNN) based vector-to-vector regression. The goal of this work is two-fold: (i) presenting performance bounds of MAE, and (ii) demonstrating new properties of MAE that make it more appropriate than mean squared error (MSE) as a loss function for DNN based vector-to-vector regression. First, we show that a generalized upper-bound for DNN-based vector-to-vector regression can be ensured by leveraging the known Lipschitz continuity property of MAE. Next, we derive a new generalized upper bound in the presence of additive noise. Finally, in contrast to conventional MSE commonly adopted to approximate Gaussian errors for regression, we show that MAE can be interpreted as an error modeled by Laplacian distribution. Speech enhancement experiments are conducted to corroborate our proposed theorems and validate the performance advantages of MAE over MSE for DNN based regression.
\end{abstract}


\IEEEpeerreviewmaketitle

\section{Introduction}
\label{sec1}
\IEEEPARstart{M}{ean} absolute error (MAE), originated from a measure of average error~\cite{willmott1985statistics}, is often employed in assessing vector-to-vector (a.k.a. multivariate) regression models \cite{borchani2015}. Another form of average error is a root-mean-squared error (RMSE), but MAE was shown to outperform RMSE for measuring an average model accuracy in most situations except the Gaussian noisy scenarios~\cite{chai2014root, willmott2005advantages, willmott2009ambiguities}. An exception occurs when the expected error satisfies Gaussian-distributed and enough training samples are available~\cite{chai2014root}. Besides, mean squared error (MSE) is the squared form of RMSE and is commonly adopted as a regression loss function~\cite{hestie2001, bishop:2006, mohri2018foundations, shalev2014understanding}.

In the literature, there are some discussions on the relationship between MSE and MAE. Berger~\cite{Berger:1985} presented pros and cons of squared and absolute errors from an estimation point of view. In \cite{vapnik95}, a better solution to support vector machines could be obtained based on a loss function of an absolute difference instead of the quadratic error. Li et al.\cite{ning_li:2020} discussed the effectiveness of MAE and its variations when training a deep model for energy load forecasting; Imani et al.~\cite{imani2018improving} investigated distributional losses, including both MAE and MSE, for regression problems from the perspective of efficient optimization. Pandey and Wang~\cite{pandey2018adversarial} exploited the MAE and MSE loss functions for generative adversarial nets (GANs). However, a comparison between MAE and MSE in terms of generalization capabilities \cite{vapnik1971, charles2017stability, qi2019theory} is still missing in theory. Thus, this paper aims at bridging this gap. In particular, we investigate MAE and MSE in terms of performance error bounds and robustness against various noises in the context of the deep neural network (DNN) based vector-to-vector regression, since DNNs offer better representation power and generalization capability in large-scale regression problems, such as those addressed in \cite{lorencs2008, Takeda:2007, xu2015regression, qi2020tensor}.

In this paper, we prove that the Lipschitz continuity property~\cite{mangasarian1987lipschitz, o2006metric}, which holds for MAE but not for MSE, is a necessary condition to derive the upper bound on the Rademacher complexity \cite{bartlett2002rademacher, bartlett2005} of DNN based vector-to-vector regression functions, as we have demonstrated in~\cite{qi2020theory}. Next, we show that the MAE Lipschitz continuity property can also result in a new upper bound on the generalization capability of DNN-based vector-to-vector regression in the presence of additive noise~\cite{Su2019, yang2020characterizing, weng2018evaluating}. Moreover, another contribution of this work is that we establish a connection between the MAE loss function and Laplacian distribution~\cite{kotz2012laplace}, which is in contrast to the MSE loss function associated with Gaussian distribution~\cite{goodman1963statistical}. In doing so, we can highlight the key advantages of MAE over MSE by comparing the characteristics of those two distributions. 

Our experiments of speech enhancement are used as the regression task to assess our theoretical derivations and empirically verify the effectiveness of MAE over MSE. We choose regression-based speech enhancement because it is an unbounded mapping from $\mathbb{R}^{d} \rightarrow \mathbb{R}^{q}$, where enhanced speech features are expected to closely approximate the clean speech features in regression. 

The remainder of this paper is presented as follows: Section~\ref{sec2} introduces some necessary math notations and theorems. Sections~\ref{sec3}, and ~\ref{sec4} highlight some key properties of the MAE loss function for DNN based vector-to-vector regression. Section~\ref{sec5} associates the MAE loss function with the Laplacian distribution. The related experiments of speech enhancement are given in Section~\ref{sec6}, and Section~\ref{sec7} concludes this work.

\section{Preliminaries}
\label{sec2}

1. \textbf{Notations}

\begin{itemize}
\item $f\circ g$:  The composition of functions $f$ and $g$.
\item $||\textbf{x}||_{p}$: $L_{p}$ norm of the vector $\textbf{x}$. 	
\item $\mathbb{R}^{d}$: $d$-dimensional real coordinate space. 				
\item $[n]$: An integer set $\{ 1, 2, ..., n\}$.
\item $\textbf{1}$: Vector of all ones.
\end{itemize}

\vspace{1.5mm}

\noindent 2. \textbf{Lipschitz continuity}
\begin{definition}
A function $f$ is $\beta$-Lipschitz continuous if for any $\textbf{x}, \textbf{y} \in \mathbb{R}^{n}$, for an integer $p \ge 1$, 
\begin{equation}
|| f(\textbf{x}) - f(\textbf{y}) ||_{p} \le \beta ||\textbf{x} - \textbf{y} ||_{p}.
\end{equation}
\end{definition}

\vspace{1.5mm}

\noindent 3. \textbf{Mean Absolute Error (MAE)}
\begin{definition}
MAE measures the average magnitude of absolute differences between $N$ predicted vectors $S = \{ \textbf{x}_{1}, \textbf{x}_{2}, ..., \textbf{x}_{N}\}$ and $S^{*} = \{ \textbf{y}_{1}, \textbf{y}_{2}, ..., \textbf{y}_{N}\}$, the corresponding loss function is defined as:
\begin{equation}
\mathcal{L}_{MAE}(S, S^{*}) = \frac{1}{N} \sum\limits_{i=1}^{N} || \textbf{x}_{i} - \textbf{y}_{i} ||_{1},
\end{equation}
where $||\cdot||_{1}$ denotes $L_{1}$ norm.
\end{definition}

\vspace{1.5mm}

\noindent 4. \textbf{Mean Squared Error (MSE)}
\begin{definition}
MSE denotes a quadratic scoring rule that measures the average magnitude of $N$ predicted vectors $S = \{ \textbf{x}_{1}, \textbf{x}_{2}, ..., \textbf{x}_{N} \}$ and $N$ actual observations $S^{*} = \{\textbf{y}_{1}, \textbf{y}_{2}, ..., \textbf{y}_{N}\}$, the corresponding loss function is shown as:
\begin{equation}
\mathcal{L}_{MSE}(S, S^{*}) = \frac{1}{N} \sum\limits_{i=1}^{N} ||\textbf{x}_{i} - \textbf{y}_{i} ||_{2}^{2},
\end{equation}
where $||\cdot||_{2}$ denotes $L_{2}$ norm.
\end{definition}

\vspace{1.5mm}

\noindent 5. \textbf{Empirical Rademacher Complexity}
\begin{definition}
The empirical Rademacher complexity of a hypothesis space $\mathbb{H}$ of functions $h : \mathbb{R}^{n} \rightarrow \mathbb{R}$ with respect to $N$ samples $S = \{\textbf{x}_{1}, \textbf{x}_{2}, ..., \textbf{x}_{N}\}$ is:
\begin{equation}
\mathcal{\hat{R}}_{S}(\mathbb{H}) := \mathbb{E}_{\sigma_{1}, ..., \sigma_{N}} \left[\sup\limits_{h\in \mathbb{H}} \frac{1}{N} \sum\limits_{i=1}^{N} \sigma_{i} h(\textbf{x}_{i})   \right].
\end{equation}
where $\sigma_{1}, \sigma_{2}, ..., \sigma_{N}$ are the Rademacher random variables, which are defined by the uniform distribution as:
\begin{equation}
\label{eq:rrv}
\sigma_{i} = \begin{cases}
			$1$, \hspace{4mm}  \text{with probability $\frac{1}{2}$} \\
			$-1$, \hspace{3mm}  \text{with probability $\frac{1}{2}$}.
		\end{cases}
\end{equation}
\end{definition}
In \cite{fan2019selective, zhu2009human, wainwright2019high}, it was shown that a function class with larger empirical Rademacher complexity is more likely to be overfit to the training data.

\section{MAE Loss Function for Upper Bounding Empirical Rademacher Complexity}
\label{sec3}

The Lipschitz continuity property is fundamental to derive an upper bound of the estimated regression error. In the following in Lemma~\ref{mae_lip}, we show that the MAE loss function can ensure the Lipschitz continuity property. In Lemma~\ref{mse_lip}, we instead show that the property does not hold for MSE.

\begin{lemma}
\label{mae_lip}
The MAE loss function is $1$-Lipschitz continuous.
\end{lemma}
\begin{proof}
For two vectors $\textbf{x}_{1}, \textbf{x}_{2} \in \mathbb{R}^{q}$, and a target vector $\textbf{x} \in \mathbb{R}^{q}$, the MAE loss difference is
\begin{equation}
\begin{split}
&\hspace{4mm} \left| \mathcal{L}_{MAE}(\textbf{x}_{1}, \textbf{x})  -  \mathcal{L}_{MAE}(\textbf{x}_{2}, \textbf{x} ) \right| \\
&= |||\textbf{x}_{1} - \textbf{x}||_{1} - ||\textbf{x}_{2} - \textbf{x}||_{1} |	\\
&\le || \textbf{x}_{1} - \textbf{x}_{2} ||_{1}	\hspace{25mm} \text{(triangle inequality)}	\\
&= \mathcal{L}_{MAE}(\textbf{x}_{1}, \textbf{x}_{2}). 
\end{split}
\end{equation}
\end{proof}

\begin{lemma}
\label{mse_lip}
The MSE loss function cannot lead to the Lipschitz continuity property.
\end{lemma}

\begin{proof}
$\forall \textbf{x}_{1}, \textbf{x}_{2} \in \mathbb{R}^{q}$, and $\text{ } ||\textbf{x}_{2}||_{2}^{2} > ||\textbf{x}_{1}||_{2}^{2}$, there is
\begin{equation}
\label{eq:11}
||\textbf{x}_{1}  - \textbf{x}_{2}||_{2}^{2} = ||\textbf{x}_{1}||_{2}^{2} +||\textbf{x}_{2}||_{2}^{2}  - 2\textbf{x}_{1}^{T}\textbf{x}_{2} .
\end{equation}
Next, we assume $\textbf{x} = 2\textbf{x}_{2}$, we have that
\begin{equation}
\label{eq:12}
\begin{split}
&\hspace{5mm} ||\textbf{x}_{1} - \textbf{x}||_{2}^{2} - ||\textbf{x}_{2} - \textbf{x} ||_{2}^{2} 	\\
&= ||\textbf{x}_{1}||_{2}^{2} - 2\textbf{x}_{1}^{T}\textbf{x} - ||\textbf{x}_{2}||_{2}^{2} + 2\textbf{x}_{2}^{T}\textbf{x}	\\
&= ||\textbf{x}_{1}||_{2}^{2} - 4\textbf{x}_{1}^{T}\textbf{x}_{2} - ||\textbf{x}_{2}||_{2}^{2} + 4 ||\textbf{x}_{2}||_{2}^{2}	\\
&= ||\textbf{x}_{1}||_{2}^{2} - 4\textbf{x}_{1}^{T} \textbf{x}_{2} + 3 ||\textbf{x}_{2}||_{2}^{2 }.
\end{split}
\end{equation}

By reducing Eq. (\ref{eq:11}) from Eq. (\ref{eq:12}),
\begin{equation}
\begin{split}
&\hspace{5mm} ||\textbf{x}_{1} - \textbf{x}||_{2}^{2} - ||\textbf{x}_{2} - \textbf{x} ||_{2}^{2}  - ||\textbf{x}_{1}  - \textbf{x}_{2}||_{2}^{2} \\
&= 2||\textbf{x}_{2}||_{2}^{2} - 2\textbf{x}_{1}^{T}\textbf{x}_{2} 		\\
&> ||\textbf{x}_{2}||_{2}^{2} + ||\textbf{x}_{1}||_{2}^{2}  - 2\textbf{x}_{1}^{T} \textbf{x}_{2} \\
&= || \textbf{x}_{1} - \textbf{x}_{2} ||_{2}^{2}	\\
&>0,
\end{split}
\end{equation}
we derive that
\begin{equation}
 \left| ||\textbf{x}_{1} - \textbf{x}||_{2}^{2} - ||\textbf{x}_{2} - \textbf{x} ||_{2}^{2}  \right|  > ||\textbf{x}_{1}  - \textbf{x}_{2}||_{2}^{2}, 
\end{equation}
which contradicts the property of Lipschitz continuity. Thus, the MSE loss function is not Lipschitz continuous.
\end{proof}

We now discuss the characteristic of Lipschitz continuity derived from the MAE loss function for upper bounding the estimation error $\mathcal{T}$, which is associated with the generalization capability and defined as:
\begin{equation}
\label{eq:es}
\mathcal{T} = \sup\limits_{f_{v} \in \mathbb{F}} \left| \mathcal{L}(f_{v}) - \mathcal{\hat{L}}(f_{v})   \right| \le \mathcal{\hat{R}}_{S}(\mathbb{L}).
\end{equation}
where $\mathbb{F} =\{f_{v} : \mathbb{R}^{d} \rightarrow \mathbb{R}^{q} \}$ is a family of DNN based vector-to-vector functions and $\mathbb{L} = \{ \mathcal{L}(f_{v}, f_{v}^{*}) : \mathbb{R}^{d} \times \mathbb{R}^{d} \rightarrow \mathbb{R}, f_{v} \in \mathbb{F}\}$ denotes the family of generalized MAE loss functions. In~\cite{qi2020theory}, we have shown that the estimation error $\mathcal{T}$ can be upper bounded by the empirical Rademacher complexity $\mathcal{\hat{R}}_{S}(\mathbb{L})$.

In~\cite{qi2020theory}, we have also shown that the estimation error $\mathcal{T}$ can be further upper-bounded as:
\begin{equation}
\mathcal{T} = \sup\limits_{f_{v} \in \mathbb{F}} \left| \mathcal{L}(f_{v}) - \mathcal{\hat{L}}(f_{v})   \right| \le \mathcal{\hat{R}}_{S}(\mathbb{L}) \le \mathcal{\hat{R}}_{S}(\mathbb{F}),
\end{equation}
where $\mathcal{\hat{R}}_{S}(\mathbb{F})$ is defined as:
\begin{equation}
\mathcal{\hat{R}}_{S}(\mathbb{F}) = \frac{1}{N} \mathbb{E}_{\boldsymbol{\sigma}}\left[ \sup\limits_{f_{v} \in \mathbb{F}} \sum\limits_{i=1}^{N} (\sigma_{i} \textbf{1})^{T} f_{v}(\textbf{x}_{i})\right], 
\end{equation}
where $\boldsymbol{\sigma} = \{\sigma_{1}, \sigma_{2}, ..., \sigma_{N}\}$ denotes a set of Rademacher random variables as shown in Definition 4. 

\section{MAE loss function for DNN robustness against additive noises}
\label{sec4}

We now show that the MAE loss function can give an upper bound for regression errors to ensure DNN robustness against additive noises.

\begin{theorem}
\label{thm2}
For an objective function $h = \mathcal{L} \circ f_{v}: \mathbb{R}^{d} \rightarrow \mathbb{R}$ with the MAE loss function $\mathcal{L}: \mathbb{R}^{q} \rightarrow \mathbb{R}$ and a vector-to-vector regression function $f_{v}: \mathbb{R}^{d} \rightarrow \mathbb{R}^{q}$, the difference of the objectives for adding noise $\boldsymbol{\eta}$ to signal $\textbf{x}$ is bounded as:
\begin{equation}
\label{eq:lbound}
\left|  h(\textbf{x} + \boldsymbol{\eta}) - h(\textbf{x}) \right|  \le  L_{2} ||\boldsymbol{\eta}||_{2}, 
\end{equation}
where $L_{2} = \sum_{i=1}^{q} L_{2, i}$ is the Lipschitz constant for DNN based vector-to-vector regression, and each $L_{2, i}$ is shown as:
\begin{equation}
L_{2, i} = \sup\{ ||\nabla f_{i}(\textbf{x})||_{2}: \textbf{x} \in \mathbb{R}^{d} \}.
\end{equation}
\end{theorem}

\begin{proof}
To prove Theorem~\ref{thm2}, we first introduce Lemma~\ref{lem2}, which is achieved by the modification of Theorem 1 in \cite{paulavivcius2006analysis}.
\begin{lemma}
\label{lem2}
For a vector-to-vector regression function $f:\mathbb{R}^{d} \rightarrow \mathbb{R}^{q}$ with the property of Lipschitz continuity, $\forall \textbf{x}, \textbf{y} \in \mathbb{R}^{d}$, there exists an inequality as:
\begin{equation}
||f(\textbf{x}) - f(\textbf{y})||_{1} \le L_{p} || \textbf{x} - \textbf{y} ||_{q},
\end{equation}
where $L_{p} = \sup\{ ||\nabla f(\textbf{x})||_{p}: \textbf{x} \in \mathbb{R}^{d} \}$ is a Lipschitz constant, and $\frac{1}{p} + \frac{1}{q} = 1, p, q \ge 1$.
\end{lemma}

We employ the fact that DNNs with the ReLU activation function are Lipschitz continuous~\cite{fazlyab2019efficient}. Then, based on both triangle inequality and Lemma~\ref{lem2}, we can upper bound the difference of objective functions with and without the additive noise $\boldsymbol{\eta}$ as:
\begin{equation*}
\begin{split}
| h(\textbf{x} + \boldsymbol{\eta}) - h(\textbf{x}) | &=  \left|||f(\textbf{x} + \boldsymbol{\eta})||_{1} - ||f(\textbf{x})||_{1} \right|	\\
&\le \left|\left|f(\textbf{x} + \boldsymbol{\eta}) - f(\textbf{x}) \right|\right|_{1}	 \hspace{5mm}\text{(triangle ineq.)} \\
&= L_{2} ||\boldsymbol{\eta}||_{2}  \hspace{28.5mm}\text{(Lemma 2)} \\
\end{split}
\end{equation*}
which completes the proof.
\end{proof}

Theorem~\ref{thm2} holds for the MAE loss function but is not valid for MSE loss because it is not Lipschitz continuous. In other words, the difference of additive noises imposed upon the DNN based vector-to-vector function is unbounded on the MSE loss function but the MAE can guarantee an upper bound. 

The upper bound makes more sense when the additive noise is small because the upper bound suggests that the imposed noise cannot lead to significant performance degradation. 

\section{Connection of MAE Loss Function to Laplacian Distribution}
\label{sec5}
We now separately link the  MAE and MSE loss functions to Laplacian distribution (LD) and Gaussian distribution (GD) based loss functions as defined in~\cite{chai2019using}. Both LD and GD based losses for DNN-based multivariate regression were experimentally compared and contrasted in~\cite{chai2019using}, and it was shown that the LD loss can attain better vector-to-vector regression accuracies than those obtained optimizing GD losses.

For $N$ input samples $\{\textbf{x}_{1}, \textbf{x}_{2}, ..., \textbf{x}_{N}\}$ and $N$ target vectors $\{ \textbf{y}_{1}, \textbf{y}_{2}, ..., \textbf{y}_{N} \}$, assuming $f: \mathbb{R}^{d} \rightarrow \mathbb{R}^{q}$ is a vector-to-vector regression function, we change the MAE loss function as:
\begin{equation}
\begin{split}
\mathcal{L}_{MAE}(S, S^{*}) &= \frac{1}{N} \sum\limits_{i=1}^{N} ||f(\textbf{x}_{n}) - \textbf{y}_{n}||_{1} 	\\
&= \frac{1}{N} \sum\limits_{n=1}^{N}\sum\limits_{m=1}^{d} \left| f_{m}(\textbf{x}_{n}) - y_{n, m} \right|		\\
&= \frac{1}{N} \sum\limits_{n=1}^{N} \sum\limits_{m=1}^{d} \frac{|\hat{f}_{m}(\textbf{x}_{n}) - \hat{y}_{n, m}|}{\alpha_{m}} ,
\end{split}
\end{equation}
where $\hat{f}_{m}(\textbf{x}_{n}) = \alpha_{m} f_{m}(\textbf{x}_{n})$, $\hat{y}_{n, m} = \alpha_{m} y_{n, m}$, and $\alpha_{m}$ is the variance of dimension $m$.

To link the LD based loss function $\mathcal{L}_{LD}(S, S^{*})$ in~\cite{chai2019using}, an additional term $N\sum_{m=1}^{d} \ln \alpha_{m}$ is added to $\mathcal{L}_{MAE}(S, S^{*})$, and we obtain
\begin{equation}
\mathcal{L}_{LD}(S, S^{*}) = \mathcal{L}_{MAE}(S, S^{*}) + N\sum\limits_{m=1}^{d} \ln \alpha_{m}.
\end{equation}

Moreover, an MSE based loss function can be modified as:
\begin{equation}
 \mathcal{L}_{MSE}(S, S^{*}) = \frac{1}{N} \sum\limits_{n=1}^{N} \sum\limits_{m=1}^{d} \frac{|\hat{f}_{m}(\textbf{x}_{n}) - \hat{y}_{n, m}|^{2}}{\alpha^{2}_{m}}.
\end{equation}
Then, the GD based loss $\mathcal{L}_{GD}(S, S^{*})$ can be derived by adding the term $N\sum_{m=1}^{d} \ln \alpha_{m}$ to the MSE loss $ \mathcal{L}_{MSE}(S, S^{*})$,
\begin{equation}
\mathcal{L}_{GD}(S, S^{*}) = \mathcal{L}_{MSE}(S, S^{*}) + N\sum\limits_{m=1}^{d} \ln \alpha_{m}.
\end{equation}

We can observe that $\mathcal{L}_{MAE}(S, S^{*})$ and $ \mathcal{L}_{MSE}(S, S^{*})$ are special cases of $\mathcal{L}_{LD}(S, S^{*}) $ and $\mathcal{L}_{GD}(S, S^{*})$ without concerning the variance terms. When $\forall m\in [d]$, the variance $\alpha_{m}$ is a constant, $ \mathcal{L}_{LD}(S, S^{*})$ and $\mathcal{L}_{GD}(S, S^{*})$ exactly correspond to $ \mathcal{L}_{MAE}(S, S^{*})$ and $\mathcal{L}_{MSE}(S, S^{*})$, respectively.

Since the work~\cite{chai2019using} suggests that the LD based loss function can achieve better regression performance than the GD based one, we show that the MAE based loss function can also keep the advantage over the MSE when the variance related terms are the same. Our experiments of speech enhancement in Section~\ref{sec6}, where both MAE and MSE loss functions are involved, are used to verify that.

\section{Experiments}
\label{sec6}
This section presents our speech enhancement experiments to corroborate the aforementioned theorems. The goal of the experiments is to verify that MAE can achieve better regression performance than  MSE under various noisy conditions because of the ensured upper bounds on the MAE loss functions for DNN-based vector-to-vector regression.

\subsection{Data Preparation}
Our experiments were conducted on the Edinburgh noisy speech database, where there were a total $23075$ and $824$ clean utterances for training and testing, respectively. The noisy training dataset at four SNR levels (15 dB, 10 dB, 5 dB, 0 dB), was obtained using the following noises: a domestic noise (inside a kitchen), an office noise (in a meeting room), three public space noises (cafeteria, restaurant, subway station), two transportation noises (car and metro) and a street noise (busy traffic intersection). In sum, we had 40 different noisy types to synthesize $23075$ noisy training speech utterances. As for the noisy test set, the noisy conditions include a domestic (living room), an office noise (office space), one transport (bus) and two street noises (open area cafeteria and a public square) at four SNR values (17.5 dB, 12.5 dB, 7.5 dB, 2.5 dB). Thus, there were $20$ various noisy conditions for generating totally $824$ noisy test speech utterances. The Edinburgh noisy speech corpus provides a more challenging speech scenario, which allows us to better support our Theorems.

\subsection{Experimental Setup}
In this work, DNN based vector-to-vector regression models followed feed-forward architectures, where the inputs were normalized log-power spectral (LPS) feature vectors of noisy speech~\cite{deng2004enhancement, qi2013auditory}, and the outputs were LPS features of either clean or enhanced speech. At training time, clean LPS vectors were assigned to the top layer of DNN to function as targets. At test time, the top layer of DNN generated enhanced LPS vectors. The architecture of DNN had the structure $771$-$800$-$800$-$800$-$800$-$800$-$1600$-$257$, which corresponds to Input$-$Hidden$-$Output. The ReLU activation function was employed in the hidden neurons, and the top layer was connected to a linear function for vector-to-vector regression. The enhanced waveforms were reconstructed based on the overlap-add method as shown in \cite{xu2015regression}. The technique of global variance equalization~\cite{silen2012ways} was utilized to improve the subjective perception of speech enhancement. At training time, the standard back-propagation (BP) algorithm was adopted to update the model parameters. The MAE and MSE loss functions were separately used to measure the difference between normalized LPS features and the reference ones. The stochastic gradient descent (SGD) based optimizer with a learning rate of $1 \times 10^{-3}$ and a momentum rate of $0.4$ was set up for the BP algorithm. Moreover, noise-aware training (NAT)~\cite{xu2014dynamic} was also used to enable non-stationary noise awareness. Context information was accounted at the input by using $3$ LPS vectors by concatenating frames within a sliding window~\cite{feature3, qi2013bottleneck, qi2016robust}. During the training time, the maximum $20$ epochs were set, and one-tenth of training data were randomly split into a validation set. If the performance of the model on the validation dataset started to degrade, the training process was stopped. 

The evaluation metrics were based on three types of criteria: MAE, MSE, perceptual evaluation of speech quality (PESQ)~\cite{rix2001perceptual}, and short-time objective intelligibility (STOI)~\cite{taal2010short}. PESQ, which ranges from $-0.5$ to $4.5$, is an indirect evaluation which is highly correlated with speech quality. A higher PESQ score corresponds to a better perception quality. Similarly, the STOI score, which ranges from 0 to 1, also refers to a measurement of predicting the intelligibility of noisy or enhanced speech. A higher STOI score corresponds to a better speech intelligibility. 

\subsection{Evaluation Results}

Using the DNN models trained with the MAE criterion (DNN-MAE) and the MSE criterion (DNN-MSE), Table~\ref{tab:tab1} lists the MAE values for speech enhancement experiments with test data. The MAE values evaluated with DNN-MAE in the top row are always lower than those in the bottom row evaluated with DNN-MSE under the same noisy condition in each column. More specifically, MAE scores by DNN-MAE achieves a lower value than DNN-MSE (0.7812 vs. 0.8278). Similarly, DNN-MAE achieves a lower MSE score than DNN-MSE (0.7954 vs. 0.8371). Besides, the MAE scores for both DNN-MAE and DNN-MSE are consistently lower than the MSE values. 

\begin{table}[bh]\footnotesize
\center
\renewcommand{\arraystretch}{1.3}
\caption{The MAE and MSE Values on Edinburgh speech corpus.}
\label{tab:tab1}
\begin{tabular}{|c||c|c|}
\hline
Models 	   &   MAE     &    MSE	   	\\
\hline
DNN-MAE	   &	0.7812  &	  0.7954		 \\
\hline
DNN-MSE   &	0.8278  &	  0.8371		\\
\hline
\end{tabular}
\vspace{-2mm}
\end{table}

\begin{table}[bh]\footnotesize
\center
\renewcommand{\arraystretch}{1.3}
\caption{The PESQ and STOI scores on Edinburgh speech corpus.}
\label{tab:tab2}
\begin{tabular}{|c||c|c|}
\hline
Models 		&  PESQ   &  STOI     	\\
\hline
DNN-MAE		& 2.93 	 &  0.8509		\\
\hline
DNN-MSE 	& 2.85	&  0.8317		\\
\hline
\end{tabular}
\vspace{-2mm}
\end{table}

Moreover, Table~\ref{tab:tab2} shows PESQ and STOI scores obtained with the DNN-MAE and DNN-MSE models. It can be seen that the DNN model trained with the MAE criterion consistently outperforms models trained with the MSE criterion ($2.93$ vs. $2.85$ for PESQ, and $0.8509$ vs. $0.8317$ for STOI), which further confirms that MAE is a good objective function to optimize when training DNNs for speech enhancement. 

Furthermore, the performance advantages of DNN-MAE over DNN-MSE corresponds to the aforementioned theorems: (1) the upper bound in Eq.~(\ref{eq:lbound}) ensures more robust performance against the additive noise; (2) the performance gain is consistent with the connection between MAE loss function and the Laplacian distribution. 

\section{Conclusion}
This work investigates the advantages of the MAE loss function for DNN based vector-to-vector regression. On one hand, we emphasize that the Lipschitz continuity property can not only ensure a performance upper bound on DNN-based vector-to-vector regression but also give an upper bound to predict the robustness against additive noises. On the other hand, we associate the MAE loss function with Laplacian distribution. Our experiments show that DNN based regression optimized with the MAE loss function can achieve lower loss values than those obtained with the MSE counterpart. Moreover, the MAE loss function can also lead to better-enhanced speech quality in terms of the PESQ and STOI scores. Our empirical results are in line with the proposed theorems for MAE and indirectly reflect that the MAE loss functions can benefit from its related upper bounds as shown in this study.

\label{sec7}
\newpage

\bibliographystyle{IEEEtran}
\bibliography{speech}

\end{document}